\documentclass[11pt,a4paper]{article}

\usepackage{amsmath,amssymb}
\usepackage{amsfonts}
\usepackage{latexsym}
\usepackage{amsthm} 
\usepackage{color}

\theoremstyle{plain}
\newtheorem{thm}{Theorem}[section]
\newtheorem{prop}[thm]{Proposition}
\newtheorem{lem}[thm]{Lemma}
\newtheorem{cor}[thm]{Corollary}

\theoremstyle{definition}
\newtheorem{con}[thm]{Condition}

\theoremstyle{remark}

\makeatletter
\@addtoreset{equation}{section}

\makeatother

\title{Limiting absorption principle and radiation condition for repulsive Hamiltonians}
\author{Kyohei Itakura\thanks{Graduate School of Science, Kobe University, Hyogo, Japan} }
\date{}

\begin{document}

\maketitle

\begin{abstract}

For spherically symmetric repulsive Hamiltonians we prove the Besov bound, the radiation condition bounds and the limiting absorption principle.
The Sommerfeld uniqueness result also follows as a corollary of these.
In particular, the Hamiltonians considered in this paper cover the case of inverted harmonic oscillator.
In the proofs of our theorems, we mainly use a commutator argument invented recently by Ito and Skibsted.
This argument is simple and elementary, and dose not employ energy cut-offs or the microlocal analysis.

\end{abstract}

%%%%%%%%%%%%%%%%%%%%%%%%%%%%%%%%%%%%%%%%%%%%%%%%%%%%%%%%%%%%

\section{Introduction}

For any fixed $\epsilon\in (0,2]$ we consider the repulsive Schr\"odinger operator
$$
H = -\frac{1}{2}\Delta - |x|^{\epsilon} + q; \quad -\Delta = p_j\delta^{jk}p_k, \ p_j = -i\partial_{x_j},
$$
on the Hilbert space ${\mathcal H} = L^2({\mathbb R}^d)$.
Here $q$ is a real-valued function that may grow slightly slower than $|x|^\epsilon$, 
$\delta^{jk}$ is the Kronecker delta, and we use the Einstein summation convention.
Throughout the paper we will use this convention.
By the Faris-Lavine theorem (see \cite[I\hspace{-1pt}I]{rs}) 
the operator $H$ is essentially self-adjoint on $C^\infty_0(\mathbb R^d)$,
and we denote the self-adjoint extension by the same letter.
For the case $\epsilon=2$ the Hamiltonian $H$ is called the inverted harmonic oscillator.

In this paper we study properties of the resolvent
$$
R(z) = (H-z)^{-1}.
$$
We prove the Besov boundedness, the radiation condition bounds, the limiting absorption principle and the Sommerfeld uniqueness result.
The Besov boundedness yields the absence of singular continuous spectrum of $H$.
In this paper the limiting absorption principle is derived from the Besov boundedness and the radiation condition bounds.
The Sommerfeld uniqueness result characterizes the limiting resolvents by the Helmholtz equation and the radiation condition.
By using the function spaces in \eqref{a-h} below, which are somewhat different from the usual one, we can deal with also the case of inverted harmonic oscillator.

To prove the above results we apply a new commutator argument with some {\it weight inside} invented recently by \cite{is}.
A feature of this argument is a choice of the conjugate operator $A$.
As with \cite{i}, we choose $A$ to be a generator of some radial flow, not of dilations or translations.

Spectral theory for the repulsive Hamiltonians was also studied by \cite{bchm}.
However, to use the Mourre theory they introduced a new conjugate operator by using the pseudo-differential operator.
We do not use the Mourre theory or the pseudo-differential operator.
Due to this, our argument is simpler than theirs.

\subsection{Basic setting}
Choose $\chi \in C^{\infty}({\mathbb R})$ such that
\begin{equation}\label{chi}
\chi (t) =
\begin{cases}
1 & {\rm for} \ t \leq 1, \\
0 & {\rm for} \ t \geq 2,
\end{cases}
\quad \chi' \leq 0,
\end{equation}
and set $r \in C^{\infty}({\mathbb R}^d)$ and the associated differential operator $\nabla^r$ as
\begin{align*}
r(x) &= \chi(|x|)+|x|\left( 1-\chi(|x|) \right), \\
\nabla^r &= (\partial_j r)\delta^{jk}\nabla_k. \notag
\end{align*}
Moreover we introduce the function $f \in C^{\infty}(\mathbb R^d)$ and the associated differential operator $\nabla^f$ as
\begin{align}\label{f}
f(r) &=
\begin{cases}
\left( r^{1- \epsilon/2}-1 \right)/(1-\epsilon/2)+1 & {\rm for} \ 0<\epsilon <2, \\
\log{r}+1 & {\rm for} \ \epsilon = 2,
\end{cases} \\
\nabla^f &= (\partial_j f)\delta^{jk}\nabla_k. \notag
\end{align}
We note that the function $f$ is continuous with regard to $\epsilon$ and the following properties hold:
$$
r \geq 1, \quad f \geq 1, \quad \nabla^f = r^{-\epsilon/2}\nabla^r.
$$
In this paper we use the function $f$ frequently.
This is closely related to the classical orbit.
In particular, it plays an important role for the case $\epsilon=2$.
We are going to see the details of this in Subsection~1.3.

\begin{con} \label{con}
The perturbation $q$ is a real-valued function.
Moreover, there exists a splitting by real-valued functions:
$$q = q_1 + q_2 ; \quad q_1 \in C^1({\mathbb R}^d), \ q_2 \in L^{\infty}({\mathbb R}^d),$$
such that for some $\rho ,C>0$ the following bounds hold globally on ${\mathbb R}^d$:
\begin{equation*}
|q_1| \leq 
\begin{cases}
Cr^{\epsilon}f^{-\rho} & {\rm for}\ 0<\epsilon<2, \\
Cr^2f^{-1-\rho} & {\rm for}\ \epsilon=2,
\end{cases}
\quad \nabla^fq_1 \leq Cf^{-1-\rho}, \quad |q_2|\leq Cf^{-1-\rho}.
\end{equation*}
\end{con}

\vspace{2mm}
We introduce the weighted Hilbert space ${\mathcal H}_s$ for $s \in {\mathbb R}$ by
$${\mathcal H}_s = f^{-s}{\mathcal H}.$$
Note that we introduced the space $\mathcal H_s$ using the function $f$, not $r$.
Here the classical orbit is related, too.
We also denote the locally $L^2$-space by
$${\mathcal H}_{\rm loc} = L^2_{\rm loc}({\mathbb R}^d).$$
We consider $B_R = \{ f<R \}$ and the characteristic functions
\begin{equation*}
F_{\nu} =
F(B_{R_{\nu+1}}\setminus B_{R_{\nu}}),\quad R_{\nu} = 2^{\nu}, \ \nu \geq 0,
\end{equation*}
where $F(\Omega)$ denotes sharp characteristic function of a subset $\Omega \subseteq {\mathbb R}^d$.
Define the spaces $\mathcal B$, ${\mathcal B}^*$ and ${\mathcal B}^*_0$ by
\begin{equation} \label{a-h}
\begin{split}
\mathcal B &= \{ \psi \in {\mathcal H}_{\rm loc} \ | \ \|\psi\|_{\mathcal B} < \infty \}, \quad \ \|\psi\|_{\mathcal B} = \sum_{\nu \geq 0} R_{\nu}^{1/2}\| F_{\nu}\psi\|_{\mathcal H},
\\
{\mathcal B}^* &= \{ \psi \in {\mathcal H}_{\rm loc} \ | \ \|\psi\|_{{\mathcal B}^*} < \infty \}, \quad \ \|\psi\|_{{\mathcal B}^*} = \sup_{\nu \geq 0} R_{\nu}^{-1/2}\| F_{\nu}\psi\|_{\mathcal H},
\\
\mathcal B^*_0&=\{\psi\in \mathcal B^*\ |\ \lim_{\nu \to \infty} R_{\nu}^{-1/2}\| F_{\nu}\psi\|_{\mathcal H}=0\},
\end{split}
\end{equation}
respectively.
We note that ${\mathcal B}_0^*$ coincides with the closure of $C_0^{\infty}({\mathbb R}^d)$ in ${\mathcal B}^*$ and for any $s>1/2$ the following inclusion relations hold:
\begin{align}\label{1720221503}
\mathcal H_s \subsetneq \mathcal B \subsetneq \mathcal H_{1/2} \subsetneq \mathcal H \subsetneq \mathcal H_{-1/2} \subsetneq \mathcal B_0^* \subsetneq \mathcal B^* \subsetneq \mathcal H_{-s}.
\end{align}
In \cite{i} we define the spaces $\mathcal B$ and $\mathcal B^*$ using the function $r$.
However, considering the classical orbit it is natural to define the spaces using the function $f$ as above.

We introduce the conjugate operator $A$ as a maximal differential operator
\begin{equation} \label{A0}
A={\rm Re}\,p^f , \quad p^f=-i\nabla^f,
\end{equation}
with domain 
\begin{align*}
\mathcal D(A)=\{\psi\in\mathcal H\ |\ A\psi\in\mathcal H\}.
\end{align*}
The conjugate operator $A$ is self-adjoint (cf.~\cite{i}) and has the following expressions:
\begin{equation} \label{A}
A = {\rm Re}\,p^f = (p^f)^* + \frac{i}{2}(\Delta f) = p^f - \frac{i}{2}(\Delta f).
\end{equation}

By the definition of $r$, there exist $c>0, r_0 \geq 1$ such that
$$
|\nabla r| \geq c,
$$
on $\{x \in \mathbb R^d \,|\, r(x) > r_0 \}$.
We set
$$
\eta = 1- \chi(r/r_0), \ \ \tilde\eta = \eta|\nabla r|^{-2},
$$
and introduce the tensor $\ell$ as follows.
\begin{equation*} \label{ell}
\ell = \delta - \tilde\eta(\nabla r)\otimes(\nabla r).
\end{equation*}
For notational simplicity, we set
\begin{equation*} \label{h}
h = r^{-\epsilon/2-1}\left( \delta - (\nabla r)\otimes(\nabla r) + 2Cf^{-1-\rho}\delta \right).
\end{equation*}
Here we choose $C>0$ large enough so that
$$
h \geq r^{-\epsilon}f^{-1}\ell + Cr^{-\epsilon}f^{-2-\rho}\delta \geq 0,
$$
as quadratic forms on fibers of the tangent bundle of $\mathbb R^d$.
For any open subset $I \subseteq \mathbb R$ let us denote
$$
I_{\pm}= \{ z=\lambda \pm i \Gamma \in \mathbb C \ | \ \lambda \in I, \ \Gamma \in (0, 1) \},
$$  
respectively.
We also use the notation $\langle T \rangle_{\psi}=\langle \psi, T\psi \rangle$.

\subsection{Results}
\begin{thm} \label{bb}
Suppose Condition~\ref{con} and let $I \subseteq \mathbb R$ be any relatively compact open subset.
Then there exists $C>0$ such that for any $\phi=R(z)\psi$ with $z \in I_{\pm}$ and $\psi \in \mathcal B$
\begin{equation} \label{bbound}
\| \phi \|_{\mathcal B^*} + \|p^f\phi \|_{\mathcal B^*} + \langle p_jh^{jk}p_k \rangle^{1/2}_{\phi} + \| r^{-\epsilon}p_j\delta^{jk}p_k\phi \|_{\mathcal B^*}  \leq C\| \psi \|_{\mathcal B}.
\end{equation}
\end{thm}

\begin{cor} \label{sc}
Under Condition~\ref{con}, the operator $H$ has no singular continuous spectrum: $\sigma_{\rm sc}(H)=\emptyset$.
\end{cor}

To prove Theorem~\ref{bb} we use the absence of $\mathcal B_0^*$-eigenfunctions for $H$.
Since the space $\mathcal B_0^*$ of this paper is somewhat different from the one in \cite{i} for $\epsilon=2$, we state the version of Rellich's theorem using in this paper in Appendix~\ref{appen}.

The absence of eigenvalue for $H$ follows immediately from Theorem~\ref{rell}.
Therefore by combining Corollary~\ref{sc} with it we obtain that the spectrum of $H$ is purely absolutely continuous under Condition~\ref{con}.
The limiting absorption principle does not immediately follow from Besov boundedness \eqref{bbound}.
To show it we impose an additional condition and we establish radiation condition bounds.

\begin{con} \label{con2}
In addition to Condition~\ref{con}, there exist $\tau, C>0$ such that
$$
|\nabla^f q_1| \leq Cf^{-1-\tau}, \quad |\ell^{\bullet k}r^{-\epsilon/2}\nabla_k q_1| \leq Cf^{-1-\tau}.
$$
\end{con}

Now we choose a smooth decreasing function $r_{\lambda} \geq 1$ of $\lambda \in \mathbb R$ such that
\begin{equation*}
\lambda - q_1 + r^{\epsilon} > 1 \ \ \text{for}\ r \geq r_{\lambda},
\end{equation*}
and set asymptotic complex phase $a$: For $z = \lambda \pm i \Gamma \in \mathbb R \cup \mathbb R_{\pm}$
\begin{equation} \label{phasea}
a = a_z = \eta_{\lambda}|\nabla r|r^{-\epsilon/2}\sqrt{2(z-q_1+r^{\epsilon})} \pm i\frac{\epsilon}{2}|\nabla r|^2r^{-\epsilon/2-1},
\end{equation}
respectively, where $\eta_{\lambda}=1-\chi(r/r_{\lambda})$.
Here we choose the branch of square root as ${\rm Re}\,\sqrt{w} >0$ for $w \in \mathbb C \setminus (-\infty, 0]$.
Let
$$
\beta_c = {\rm min}\,\{ \rho, \epsilon', \tau, 1+\epsilon/2 \}, \quad \epsilon' =
\begin{cases}
\epsilon/(1-\epsilon/2) & {\rm for}\ 0<\epsilon<2, \\
2 & {\rm for}\ \epsilon=2.
\end{cases}
$$

\begin{thm} \label{rcb}
Suppose Condition~\ref{con2}, and let $I \subset \mathbb R$ be any relatively compact open subset.
Then for all $\beta \in [0, \beta_c)$ there exists $C>0$ such that for any $\phi = R(z)\psi$ with $\psi \in f^{-\beta}\mathcal B$ and $z \in I_{\pm}$
\begin{equation}
\| f^{\beta}(A \mp a)\phi \|_{\mathcal B^*} + \langle p_if^{2\beta}h^{ij}p_j \rangle_{\phi}^{1/2} \leq C\| f^{\beta}\psi \|_{\mathcal B},
\end{equation}
respectively.
\end{thm}

By Theorem~\ref{bb} and Theorem~\ref{rcb} we obtain the limiting absorption principle.

\begin{cor} \label{lap}
Suppose Condition~\ref{con2} and let $I \subseteq \mathbb R$ be any relatively compact open subset.
For any $s>1/2$ and $\omega \in (0, {\rm min}\left\{(2s-1)/(2s+1), \beta_c/(\beta_c+1)\right\})$ there exists $C>0$ such that for any $z, z' \in I_+$ or $z, z' \in I_-$
\begin{equation} \label{Hc}
\begin{split}
\| R(z) - R(z') \|_{\mathcal B(\mathcal H_s, \mathcal H_{-s})} &\leq C|z-z'|^{\omega}, \\
\| r^{-\epsilon/2}pR(z) - r^{-\epsilon/2}pR(z') \|_{\mathcal B(\mathcal H_s, \mathcal H_{-s})} &\leq C|z-z'|^{\omega}.
\end{split}
\end{equation}
In particular, the operators $R(z)$ and $r^{-\epsilon/2}pR(z)$ attain uniform limits as $I_{\pm} \ni z \to \lambda \in I$ in the norm topology of $\mathcal B(\mathcal H_s, \mathcal H_{-s})$, say denoted by
\begin{equation} \label{Rlim}
\begin{split}
R(\lambda \pm i0) &= \lim_{I_{\pm} \ni z \to \lambda} R(z), \\
r^{-\epsilon/2}pR(\lambda \pm i0) &= \lim_{I_{\pm} \ni z \to \lambda} r^{-\epsilon/2}pR(z),
\end{split}
\end{equation}
respectively.
These limits $R(\lambda \pm i0)$ and $r^{-\epsilon/2}pR(\lambda \pm i0)$ belong to $\mathcal B(\mathcal B, \mathcal B^*)$.
\end{cor}

Combining Theorem~\ref{rcb} and Corollary~\ref{lap} we obtain the radiation condition bounds for real spectral parameters.

\begin{cor} \label{rcb2}
Suppose Condition~\ref{con2}, and let $I \subset \mathbb R$ be any relatively compact open subset.
Then for all $\beta \in [0, \beta_c)$ there exists $C>0$ such that for any $\phi = R(\lambda \pm i0)\psi$ with $\psi \in f^{-\beta}\mathcal B$ and $z \in I_{\pm}$
\begin{equation}
\| f^{\beta}(A \mp a)\phi \|_{\mathcal B^*} + \langle p_if^{2\beta}h^{ij}p_j \rangle_{\phi}^{1/2} \leq C\| f^{\beta}\psi \|_{\mathcal B},
\end{equation}
respectively.
\end{cor}

Finally, we obtain the Sommerfeld uniqueness result.
\begin{cor} \label{Sur}
Suppose Condition~\ref{con2}, and let $\lambda \in \mathbb R$, $\phi \in \mathcal H_{\rm loc}$ and $\psi \in f^{-\beta}\mathcal B$ with $\beta \in [0, \beta_c)$.
Then $\phi = R(\lambda \pm i0)\psi$ holds if and only if both of the following conditions hold:
\vspace*{-2mm}
\begin{itemize}
  \setlength{\itemsep}{-1mm}
  \item[(i)] $(H-\lambda)\phi = \psi$ \ in the distributional sense.
  \item[(ii)] $\phi \in f^{\beta}\mathcal B^*$ and $(A \mp a)\phi \in f^{-\beta}\mathcal B_0^*$.
\end{itemize}
\end{cor}

As is seen in Appendix~\ref{appen}, there is no generalized eigenfunction in $\mathcal B_0^*$.
We constructed a $\mathcal B^*$-eigenfunction in \cite{i} (see also Appendix~\ref{appen}).
Therefore in the sense that the inclusion relations \eqref{1720221503} hold, the space $\mathcal B^*$ is the minimal space where a generalized eigenfunction exists.
Hence Theorem~\ref{bb} asserts the boundedness of $R(z)$ between natural and optimal spaces.
As far as the author knows, there seem to be no literature on the Besov boundedness for repulsive Hamiltonians so far, and our theorem is new.
In particular, by setting the spaces $\mathcal B$ and $\mathcal B^*$ using the function $f$ of \eqref{f}, even for the case $\epsilon=2$ we obtain the results.
In fact if we define the spaces $\mathcal B$ and $\mathcal B^*$ using the function $r$,
the proof of Theorem~\ref{rcb} is not completed.

To prove the theorems and the corollaries we apply a new commutator argument with some weight inside from \cite{is}.
In \cite{is}, they consider only potentials decaying at infinity.
In order to deal with the repulsive potentials that diverge to $-\infty$ at infinity we need to choose the appropriate conjugate operator $A$ as \eqref{A0}.

The limiting absorption principle for repulsive Hamiltonians was studied also by \cite{bchm}.
However, they did not prove the Besov boundedness.
Moreover, as for the decay rate of perturbation at infinity, our assumptions are considerably weaker and includes their setting.
In this sense, our results are stronger than theirs.

In case $\epsilon=0$, there has been an extensive amount of literature on spectral theory (e.g. \cite{a, fh, fhh2o, ho, ij, is, iso}).
As for the case $\epsilon=2$, Ishida studied inverse scattering problem in \cite{ishi} and borderline of the short-range condition in \cite{ishi2}.
Moreover Finster and Isidro discussed the $L^p$-spectrum in \cite{fi}.
Skibsted dealt with the Besov bound and the limiting absorption principle for attractive Hamiltonians in \cite{ski}, whereas we considered the case of repulsive Hamiltonians.
We also mention recent works related to the repulsive potentials.
Josef studied in \cite{j} the properties of spectrum of two-dimensional Pauli operator with repulsive potential.
Lakaev studied in \cite{l1,l2} eigenvalue problem for discrete Schr\"odinger operator with repulsive potential on the two-dimensional lattice $\mathbb Z^2$.

In Section~2 we introduce a commutator with weight inside and discuss its properties.
In Section~3 we prove Theorem~\ref{bb} by using a commutator estimate and contradiction.
In Section~4 by using Theorem~\ref{bb} we prove Theorem~\ref{rcb} and Corollaries~\ref{lap}-\ref{Sur}.
In the proofs of the these results commutator estimates play major roles.

%%%%%%%%%%%%%%%%%%%%%%%%%%%%%%%%%%%%%%%%%%%%%%%%%%%%%%%%%%%%

\subsection{Classical orbit}

In this subsection we consider the classical orbit on the Hamiltonian
$$
H = \frac{1}{2}p^2 - |x|^{\epsilon}.
$$
The Hamilton equation is given by
\begin{equation*}
\dot{x}(t) = p(t), \quad \dot{p}(t) = -\epsilon|x(t)|^{\epsilon-2}x(t).
\end{equation*}
This yields the following equation:
\begin{equation}\label{ddot}
\ddot{x}(t) = \epsilon|x(t)|^{\epsilon-2}x(t).
\end{equation}
As for the case $\epsilon=2$ we can compute explicitly:
\begin{equation*}
x(t) = \frac{1}{2}\left( x(0) + \frac{1}{\sqrt 2}\dot{x}(0) \right)e^{\sqrt 2t} + \frac{1}{2}\left( x(0) - \frac{1}{\sqrt 2}\dot{x}(0) \right)e^{-\sqrt 2t}.
\end{equation*}
Thus, in general, $|x(t)|$ grows exponentially as $t \to \infty$.
On the other hand for the case $0<\epsilon<2$, $|x(t)|$ grows in the order of $t^{1/(1-\epsilon/2)}$ in general.
In fact, we set $x(t) = t^{1/(1-\epsilon/2)}y$ for $y \in \mathbb R^d$ with $|y|=\left( 2^{-1}(2-\epsilon)^2 \right)^{1/(2-\epsilon)}$, and then the function $x(t)$ satisfies \eqref{ddot}.
By these observations if we define the new position function
\begin{equation*}
y(t) =
\begin{cases}
|x(t)|^{1-\epsilon/2}\left( x(t)/|x(t)| \right) & \text{for}\ \, 0<\epsilon<2, \\
\log{|x(t)|}\left( x(t)/|x(t)| \right) & \text{for}\ \, \epsilon=2,
\end{cases}
\end{equation*}
we have $|y(t)|=\mathcal O(t)$ as $t \to \infty$ similarly to the case $\epsilon=0$.
Hence it is natural to define the spaces $\mathcal B$ and $\mathcal B^*$ using the function $f$ rather than $r$.

%%%%%%%%%%%%%%%%%%%%%%%%%%%%%%%%%%%%%%%%%%%%%%%%%%%%%%%%%%%%

\section{Preliminaries}

In this section we are going to prepare some lemmas and properties to prove the results that are stated in Section~1.
For simplicity, we omit the proofs of these (cf.~\cite{i, sig}).
\begin{lem}
Let $H^2(\mathbb R^d)$ be the Sobolev space of second order, and set
$$
H^2_{\rm comp}(\mathbb R^d) = \{ \psi \in H^2(\mathbb R^d) \,|\, {\rm supp}\,\psi \ {\rm is \ compact}\}.
$$
Then the following inclusion relations hold.
\begin{equation} \label{embed}
H^2_{\mathrm{comp}}(\mathbb R^d) \subset {\mathcal D}(H) \subset {\mathcal D}(A).
\end{equation}
\end{lem}

We consider commutators with a weight $\Theta$ inside:
\begin{equation*}
[ H, iA ]_{\Theta} := i(H\Theta A - A\Theta H).
\end{equation*}
Let $\Theta =\Theta(f)$ be a non-negative smooth function with bounded derivatives.
More explicitly, if we denote its derivatives in $f$ by primes such as $\Theta'$, then
\begin{equation} \label{theta1}
\Theta \geq 0, \quad |\Theta^{(k)}|\leq C_k, \ \  k=0,1,2,\ldots .
\end{equation}
We define the quadratic form $[H,iA]_{\Theta}$ on $C^\infty_0(\mathbb R^d)$,
and then extend it to $\mathcal D(A)$ according to the following lemma.

\begin{lem} \label{lem21}
Suppose Condition \ref{con}, and let $\Theta$ be a non-negative smooth function with bounded derivatives {\rm \eqref{theta1}}.
Then, as quadratic forms on $C_0^{\infty}({\mathbb R}^d)$,
\begin{align*}
\begin{split} \label{commest0}
[H, iA]_{\Theta} &= p_j(\nabla^2 f)^{jk}\Theta p_k + (p^f)^*\Theta'p^f + \frac{1}{2}{\rm Re}\left( (\Delta f)\Theta p_i\delta^{ij}p_j \right) \\
  &\quad \, - \frac{1}{2}p_i(\Delta f)\Theta\delta^{ij}p_j - \frac{1}{2}{\rm Im}\left( (\nabla |\nabla f|^2)^j\Theta'p_j \right) - {\rm Im}\left( 2q_2\Theta p^f \right) \\
  &\quad \, - {\rm Re}\left( |\nabla f|^2\Theta' H \right) + \epsilon (\nabla f)^k|x|^{\epsilon-2}x_k\Theta + q_{\Theta} - \frac{1}{4}|\nabla f|^4\Theta''';
\end{split} \\
q_{\Theta} &= - (\nabla^f q_1)\Theta + q_2(\Delta f)\Theta + |\nabla f|^2q_2\Theta' \notag \\
  &\quad \, - \frac{1}{4}(\nabla^f |\nabla f|^2)\Theta'' - \frac{1}{4}|\nabla f|^2(\Delta f)\Theta''. \notag
\end{align*}
In particular noting the formulae \eqref{rx} below and using the Cauchy-Schwarz inequality $[H, iA]_{\Theta}$ restricted to $C_0^{\infty}({\mathbb R}^d)$ extends to a bounded form on $\mathcal D(A)$.
Here, we regard $\mathcal D(A)$ as the Banach space with graph norm.
\end{lem}

We have the following formulae (cf. \eqref{f}): for $r \geq 2$
\begin{equation}
\begin{split}\label{rx}
|\nabla r|^2&=1, \qquad (\nabla^2 f)^{jk}=r^{-\epsilon/2-1}\delta^{jk}-\left( \frac{\epsilon}{2}+1 \right)r^{-\epsilon/2-1}(\nabla r)^j(\nabla r)^k, \\
(\nabla r)^j&=x^jr^{-1}, \qquad \Delta f=(d-\frac{\epsilon}{2}-1)r^{-\epsilon/2-1}, \qquad \Delta r =(d-1)r^{-1}.
\end{split}
\end{equation}

On the other hand, throughout the paper we shall use the notation
\begin{equation*}
{\rm Im}(A\Theta H)=\frac{1}{2i}(A\Theta H - H\Theta A)
\end{equation*}
as a quadratic form defined on ${\mathcal D}(H)$, i.e. for $\psi \in {\mathcal D}(H)$
$$\langle {\rm Im}(A\Theta H) \rangle_{\psi} = \frac{1}{2i}\left( \langle A\psi, \Theta H\psi \rangle - \langle H\psi, \Theta A\psi \rangle \right).$$
Note that by the embedding \eqref{embed} the above quadratic form is well-defined.
Obviously the quadratic forms $[H, iA]_{\Theta}$ and $2{\rm Im}(A\Theta H)$ coincide on $C_0^{\infty}(\mathbb R^d)$,
and hence we obtain
\begin{equation}\label{equal}
[H, iA]_{\Theta} = 2{\rm Im}(A\Theta H) \quad {\rm on} \ {\mathcal D}(H).
\end{equation}

Finally using the function $\chi$ of \eqref{chi} we define $\chi_n, \bar\chi_n$ for $n \geq 0$ by
$$
\chi_n = \chi(f/R_n), \quad \bar\chi_n = 1 - \chi_n.
$$

%%%%%%%%%%%%%%%%%%%%%%%%%%%%%%%%%%%%%%%%%%%%%%%%%%%%%%%%%%%%

\section{Besov bound}

In this section we discuss the locally uniform Besov bound for the resolvent $R(z)$.
Lemma~\ref{bblem} and Proposition~\ref{bbprop} in Subsection~3.1 is a key to prove Theorem~\ref{bb}.
In Subsection~3.2 we prove Theorem~\ref{bb} by Proposition~\ref{bbprop} and contradiction.

\subsection{Commutator estimate}

We introduce the regularized weight
\begin{equation} \label{theta}
\Theta= \Theta_{\nu}^{\delta}=\int_0^{f/R_{\nu}}(1+s)^{-1-\delta}ds=\left[ 1-(1+f/R_{\nu})^{-\delta} \right]/\delta; \ \delta>0, \nu\geq 0,
\end{equation}
and compute derivatives  in $f$:
\begin{equation} \label{thetabibun}
\Theta'=(1+f/R_{\nu})^{-1-\delta}/R_{\nu}, \ \Theta''=-(1+\delta)(1+f/R_{\nu})^{-2-\delta}/R_{\nu}^2.
\end{equation}

\begin{prop} \label{bbprop}
Suppose Condition~\ref{con}, let $I \subseteq \mathbb R$ be any relatively compact open subset, and fix any $\delta \in (0, \min\{ 1, \rho, \epsilon' \})$ in the definition \eqref{theta} of $\Theta$.
Then there exist $C>0$ and $n \geq 0$ such that for all $\phi=R(z)\psi$ with $z \in I_{\pm}$ and $\psi \in \mathcal B$ and for all $\nu \geq 0$
\begin{equation}
\begin{split}
  &\quad \hspace{-3mm} \| \Theta'^{1/2}\phi \|^2 + \| \Theta'^{1/2}A\phi \|^2 + \langle p_jh^{jk}\Theta p_k \rangle_{\phi} \\
  &\leq C\left( \| \phi \|_{\mathcal B^*} \| \psi \|_{\mathcal B} + \| A\phi \|_{\mathcal B^*} \| \psi \|_{\mathcal B} + \| \chi_n \Theta^{1/2} \phi \|^2 \right).
\end{split}
\end{equation}
\end{prop}

We first note that $\Theta$ defined by \eqref{theta} has following properties.
\begin{lem} \label{thetaest}
Suppose Condition~\ref{con}, and fix any $\delta>0$ in the definition \eqref{theta} of $\Theta$.
Then there exist $c, C, C_k>0, \, k=2, 3, \ldots$, such that for any $k=2, 3, \ldots$ and uniformly in $\nu \geq 0$
\begin{align*}
&\quad \hspace{-3mm} c/R_{\nu} \leq \Theta \leq \min\{ C, f/R_{\nu} \}, \\
&\quad \hspace{-3mm} c\left( \min\{ R_{\nu}, f \} \right)^{\delta}f^{-1-\delta}\Theta \leq \Theta' \leq f^{-1}\Theta, \\
&\quad \hspace{-3mm} 0 \leq (-1)^{k-1}\Theta^{(k)} \leq C_kf^{-k}\Theta.
\end{align*}
\end{lem}
\begin{proof}
By the definition of $\Theta$ in \eqref{theta} and expressions of derivatives of it as \eqref{thetabibun}, 
the asserted estimates are clearly hold except for the first estimate in the second line.
But this estimate follows by using the last estimate of the first line and the following inequality:
\begin{equation*}
\bigl( {\rm min}\{ R_{\nu}, f \} \bigr)^{\delta}f^{-1-\delta}\bigl( {\rm min}\{ R_{\nu}, f \}/R_{\nu} \bigr)\left( (1+f/R_{\nu})^{1+\delta}R_{\nu} \right) \leq C.
\end{equation*}
\end{proof}
The following lemma is a key to prove Theorem~\ref{bb}.
\begin{lem} \label{bblem}
Suppose Condition~\ref{con}, let $I \subseteq \mathbb R$ be any relatively compact open subset, and fix any $\delta \in (0, \min\{1, \rho, \epsilon' \})$ in the definition \eqref{theta} of $\Theta$.
Then there exist $c, C>0$ and $n \geq 0$ such that uniformly in $z \in I_{\pm}$ and $\nu \geq 0$, as a quadratic forms on $\mathcal D(H)$,
\begin{equation}
\begin{split}
  &\quad \hspace{-3mm} {\rm Im}\bigl( A\Theta(H-z) \bigr) \\
  &\geq c\Theta' +cA\Theta'A +cp_jh^{jk}\Theta p_k - C\chi_n^2\Theta - {\rm Re}\bigl( \gamma(H-z) \bigr),
\end{split}
\end{equation} 
where $\gamma=\gamma_{z,\nu}$ is a uniformly bounded complex-valued function: $|\gamma| \leq C$.
\end{lem}

\begin{proof}
Let $I$ and $\delta$ be as in the assertion.
First using Lemmas~\ref{lem21},~\ref{thetaest}, \eqref{rx}, \eqref{equal}, the Cauchy-Schwarz inequality and a general identity holding for any $g \in C^{\infty}({\mathbb R}^d)$:
\begin{equation} \label{formula}
p_ig\delta^{ij}p_j = {\rm Re}\left( gp_i\delta^{ij}p_j \right) + \frac{1}{2}(\Delta g),
\end{equation}
we can bound uniformly in $z=\lambda \pm i \Gamma \in I_{\pm}$ and $\nu \geq 0$
\begin{equation} \label{zest}
\begin{split}
  &\quad \hspace{-3mm} {\rm Im}\bigl( A\Theta(H-z) \bigr) \\
  &\geq \frac{1}{2}p_jr^{-\epsilon/2-1}\delta^{jk}\Theta p_k - \left( \frac{\epsilon}{4}+\frac{1}{2} \right)(p^r)^*r^{-\epsilon/2-1}\Theta p^r + \frac{1}{2}(p^f)^*\Theta' p^f \\
  &\quad -\frac{1}{2}{\rm Re}\left( |\nabla f|^2\Theta'(H-z) \right) \mp \Gamma{\rm Re}(A\Theta) + \frac{\epsilon}{2}r^{\epsilon/2-1}\Theta - C_1Q \\
  &\geq \frac{1}{2}p_jr^{-\epsilon/2-1}\left( \delta^{jk} - (\nabla r)^j(\nabla r)^k + 2Cf^{-1-\rho}\delta^{jk} \right)\Theta p_k + \frac{1}{2}(p^f)^*\Theta' p^f \\
  &\quad - {\rm Re}\left( \left( \frac{1}{2}|\nabla f|^2\Theta' + \frac{\epsilon}{2}r^{-\epsilon/2-1}\Theta \right)(H-z) \right) \mp \Gamma \Theta^{1/2}A\Theta^{1/2} - C_2Q \\
  &\geq \frac{1}{4}p_jh^{jk}\Theta p_k + \frac{1}{2}\Theta' + \frac{1}{4}A\Theta'A \mp \Gamma \Theta^{1/2}A\Theta^{1/2} \\
  &\quad - {\rm Re}\left( \left( \frac{1}{2}|\nabla f|^2\Theta' + \frac{\epsilon}{2}r^{-\epsilon/2-1}\Theta - \frac{1}{2}r^{-\epsilon}\Theta' \right)(H-z) \right) - C_3Q,
\end{split}
\end{equation}
where
\begin{equation*}
Q = f^{-1-{\rm min}\{1, \rho, \epsilon'\}}\Theta + p_jr^{-\epsilon}f^{-1-{\rm min}\{1, \rho, \epsilon'\}}\Theta\delta^{jk}p_k.
\end{equation*}

To the fourth term of \eqref{zest} we apply the Cauchy-Schwarz inequality, Lemma~\ref{thetaest} and the general identity holding for any real functions $g, h \in C^1(\mathbb R^d)$:
\begin{equation*}
h{\rm Re}\left( gp^2 \right)h = {\rm Re}\left( h^2gp^2 \right) + \frac{1}{2}(\partial_j h)\delta^{jk}(\partial_k gh).
\end{equation*}
Then it follows that
\begin{equation}
\begin{split} \label{zest2}
\mp \Gamma \Theta^{1/2}A\Theta^{1/2}
  &= \mp \Gamma\Theta^{1/2}\left( {\rm Re}\,p^f \right)\Theta^{1/2} \\
  &\geq -C_4\Gamma \Theta^{1/2}r^{-\epsilon/2}(H-\lambda)r^{-\epsilon/2}\Theta^{1/2} - C_5\Gamma \\
  &= -C_4\Gamma{\rm Re}\left( \Theta^{1/2}r^{-\epsilon/2}(H-z)r^{-\epsilon/2}\Theta^{1/2} \right) \pm C_5{\rm Im}(H-z) \\
  &\geq - {\rm Re}\left( (C_4\Gamma r^{-\epsilon}\Theta \pm iC_5)(H-z) \right) - C_6Q.
\end{split}
\end{equation}
We substitute the estimate \eqref{zest2} into \eqref{zest}, and obtain
\begin{equation}
\begin{split} \label{zest3}
  &\quad \hspace{-3mm} {\rm Im}\bigl( A\Theta(H-z) \bigr) \\
  &\geq \frac{1}{4}p_jh^{jk}\Theta p_k + \frac{1}{2}\Theta' + \frac{1}{4}A\Theta'A - C_7Q \\
  &\quad - {\rm Re}\left( \left( \frac{1}{2}\left( |\nabla f|^2-r^{-\epsilon} \right)\Theta' + \frac{\epsilon}{2}r^{-\epsilon/2-1}\Theta + C_4\Gamma r^{-\epsilon}\Theta \pm iC_5 \right)(H-z) \right).
\end{split}
\end{equation}
Using \eqref{formula} and Lemma~\ref{theta} we can combine and estimate the second and fourth terms of \eqref{zest3}: For large $n \geq 0$
\begin{equation} \label{zest4}
\frac{1}{2}\Theta' - C_7Q \geq \frac{1}{4}\Theta' - C_8\chi_n^2\Theta - 2C_7{\rm Re}\left( r^{-\epsilon}f^{-1-{\rm min}\{1, \rho, \epsilon'\}}\Theta(H-z) \right).
\end{equation}
Hence by \eqref{zest3} and \eqref{zest4}, if we set
$$
\gamma = \frac{1}{2}\left(|\nabla f|^2-r^{-\epsilon} \right)\Theta' + \frac{\epsilon}{2}r^{-\epsilon/2-1}\Theta + C_4\Gamma r^{-\epsilon}\Theta \pm iC_5 - 2C_7r^{-\epsilon}f^{-1-{\rm min}\{1, \rho, \epsilon'\}}\Theta,
$$
then the assertion follows.
\end{proof}

\begin{proof}[Proof of Proposition~\ref{bbprop}]
The assertion follows immediately from Lemma~\ref{bblem}.
\end{proof}

\subsection{Besov boundedness}

Now we prove Theorem~\ref{bb} by Proposition~\ref{bbprop} and contradiction.
\begin{proof}[Proof of Theorem~\ref{bb}]
Let $I \subset \mathbb R$ be any relatively compact open subset.
We prove the assertion only for the upper sign.

\vspace{1mm}
\noindent
{\it Step 1.}
We assume that for $C_1>0$ large enough
\begin{equation} \label{besovest}
\| \phi \|_{\mathcal B^*} \leq C_1 \| \psi \|_{\mathcal B},
\end{equation}
then the bound \eqref{bbound} holds.
In fact, the last term on the left-hand side of \eqref{bbound} clearly satisfies the desired estimate by the identity
$$
r^{-\epsilon}p_j\delta^{jk}p_k\phi = r^{-\epsilon}\psi + r^{-\epsilon}(|x|^{\epsilon}-q+z)\phi
$$
and Condition~\ref{con}.
Hence it suffice to consider the second and third terms of \eqref{bbound}.
Fix any $\delta \in (0, \min\{1, \rho, \epsilon' \})$.
Then by Proposition~\ref{bbprop} and \eqref{besovest} there exists $C_2>0$ such that for any $\phi = R(z)\psi$ with $z \in I_+$ and $\psi \in \mathcal B$ uniformly in $\epsilon_1 \in (0, 1)$ and $\nu \geq 0$
\begin{equation} \label{epsilon1}
\| \Theta'^{1/2}A\phi \|^2 + \langle p_jh^{jk}\Theta p_k \rangle_{\phi} \leq \epsilon_1\| A\phi \|_{\mathcal B^*}^2 + \epsilon_1^{-1}C_2\| \psi \|_{\mathcal B}^2.
\end{equation}
In the first term on the left-hand side of \eqref{epsilon1} for each $\nu \geq 0$, noting the expression of $\Theta'$ in \eqref{thetaest}, we restrict the integral region to $B_{R_{\nu+1}} \setminus B_{R_{\nu}}$.
As for the second term on the same side we look at the estimate \eqref{epsilon1} for any fixed $\nu \geq 0$, say $\nu=0$.
Then we have the following inequality.
$$
c_1\| A\phi \|_{\mathcal B^*}^2 + c_1\langle p_jh^{jk}\Theta p_k \rangle_{\phi} \leq 2\epsilon_1\| A\phi \|_{\mathcal B^*}^2 + 2\epsilon_1^{-1}C_2\| \psi \|_{\mathcal B}^2.
$$
If we let $\epsilon_1 \in (0, c_1/2)$, the rest of \eqref{bbound} follows from this estimate and \eqref{A}.
Hence \eqref{bbound} reduces to \eqref{besovest}.

\vspace{1mm}
\noindent
{\it Step 2.}
We prove \eqref{besovest} by contradiction.
Assume the opposite, and let $z_k \in I_+$ and $\psi_k \in \mathcal B$ be such that
\begin{equation} \label{psik}
\lim_{k \to \infty}\| \psi_k \|_{\mathcal B}=0, \quad \| \phi_k \|_{\mathcal B^*}=1; \quad \phi_k =R(z_k)\psi_k. 
\end{equation}
Note that then it automatically follows that
\begin{equation} \label{phikbibun}
\| r^{-\epsilon/2}p\phi_k \|_{\mathcal B^*} + \| r^{-\epsilon}p^2\phi_k \|_{\mathcal B^*} \leq C_3.
\end{equation}
In fact, arguing similarly to Step 1, we can deduce from \eqref{psik} and Proposition~\ref{bbprop} that
$$
\| A\phi_k \|_{\mathcal B^*}^2 + \langle p_jh^{jk}p_k \rangle_{\phi_k} \leq C_4, \quad p^2\phi_k = \psi_k + (|x|^{\epsilon}-q+z)\phi_k,
$$
and these combined Condition~\ref{con}, \eqref{A} and \eqref{psik} imply \eqref{phikbibun}.
Now, choosing a subsequence and retaking $I \subseteq \mathbb R$ slightly larger, we may assume that $z_k \in I_+$ converges to some $z \in I \cup I_+$.
If the limit $z$ belongs to $I_+$, the bounds
$$
\| \phi_k \|_{\mathcal B^*} \leq \| \phi_k \|_{\mathcal H} \leq \| R(z_k) \|_{\mathcal B(\mathcal H)}\| \psi_k \|_{\mathcal H} \leq \| R(z_k) \|_{\mathcal B(\mathcal H)}\| \psi_k \|_{\mathcal B}
$$
and \eqref{psik} contradict the norm continuity of $R(z) \in \mathcal B(\mathcal H)$ in $z \in I_+$.
Hence we have the limit
\begin{equation} \label{limzk}
\lim_{k \to \infty}z_k = z = \lambda \in I.
\end{equation}
Let $s > 1/2$.
By choosing a further subsequence we may assume that $\phi_k$ converges weakly to some $\phi \in \mathcal H_{-s}$.
But then $\phi_k$ actually converges strongly in $\mathcal H_{-s}$.
To see this let us fix $s' \in (1/2, s)$ and $g \in C_0^{\infty}(\mathbb R)$ with $g=1$ on a neighborhood of $I$, and decompose for any $n \geq 0$
\begin{align*}
f^{-s}\phi_k &= f^{-s}g(H)(\chi_n f^s)(f^{-s}\phi_k) + (f^{-s}g(H)f^s)(\bar\chi_n f^{s'-s})(f^{-s'}\phi_k) \\
  &\quad + f^{-s}(1-g(H))R(z_k)\psi_k.
\end{align*}
The last term on the right-hand side converges to $0$ in $\mathcal H$ due to \eqref{psik}, 
and the second term can be taken arbitrarily small in $\mathcal H$ by choosing $n \geq 0$ sufficiently large since $f^{-s}g(H)f^s$ is a bounded operator.
By the compactness of $f^{-s}g(H)$, for fixed $n \geq 0$ the first term converges strongly in $\mathcal H$.
Therefore $\phi_k$ converges to $\phi$ in $\mathcal H_{-s}$, i.e.
\begin{equation} \label{limphik}
\lim_{k \to \infty}\phi_k = \phi \ \ \ \text{in}\ \mathcal H_{-s}.
\end{equation}
By \eqref{psik}, \eqref{limzk} and \eqref{limphik} it follows that
\begin{equation} \label{1708011338}
(H-\lambda)\phi = 0 \ \, \text{in the distributional sense}.
\end{equation}
In addition, we can verify $\phi \in \mathcal B_0^*$.
In fact, let us apply Proposition~\ref{bbprop} with $\delta=2s-1>0$ to $\phi_k=R(z_k)\psi_k$, and take the limit $k \to \infty$ using \eqref{psik}, \eqref{phikbibun}, \eqref{limphik} and Lemma~\ref{theta}.
We obtain for all $\nu \geq 0$
\begin{equation} \label{1708011404}
\| \Theta'^{1/2}\phi \| \leq \| \chi_n\Theta^{1/2}\phi \| \leq C_5R_{\nu}^{-1/2}\| \chi_n f^{1/2}\phi \|.
\end{equation}
Letting $\nu \to \infty$ in \eqref{1708011404}, we obtain $\phi \in \mathcal B_0^*$, and then we conclude $\phi=0$ by \eqref{1708011338} and Theorem~\ref{rell}.
But this is a contradiction, because similarly to Step 1 we have
$$
1=\| \phi_k \|_{\mathcal B^*}^2 \leq C_6\left( \| \psi_k \|_{\mathcal B} + \| \chi_n \phi_k \|^2 \right),
$$
and, as $k \to \infty$, the right-hand side converges to $0$.
Hence \eqref{besovest} holds.
\end{proof}

%%%%%%%%%%%%%%%%%%%%%%%%%%%%%%%%%%%%%%%%%%%%%%%%%%%%%%%%%%%%

\section{Radiation condition}

Our main purpose in this section is to prove the radiation condition bounds for complex spectral parameters.
In Subsection~4.1 we state and prove the key lemma to prove Theorem~\ref{rcb}.
In Subsection~4.2 we prove Theorem~\ref{rcb}. 
Corollaries~\ref{lap}-\ref{Sur} are also proved in the same subsection.

Throughout the section we suppose Condition~\ref{con2}, and prove the statements only for the upper sign for simplicity.

\subsection{Commutator estimate}

We introduced the conjugate operator $B$ as a maximal differential operator
\begin{equation*} \label{B}
B = {\rm Re}\,p^r = \frac{1}{2}\left( p^r + (p^r)^* \right),
\end{equation*}
with domain
$$
\mathcal D(B) = \{ \psi \in \mathcal H \ | \ B\psi \in \mathcal H \},
$$
and set associated asymptotic complex phase $b$: For $z = \lambda + i \Gamma \in \mathbb R \cup \mathbb R_+$
\begin{equation} \label{phaseb}
b = b_z = \eta_{\lambda}|\nabla r|\sqrt{2(z-q_1+r^{\epsilon})} + i\frac{\epsilon}{4}|\nabla r|^2r^{-1}.
\end{equation}
We note that the operator $B$ is self-adjoint on $\mathcal D(B)$ (cf.~\cite{is}). 

\begin{lem} \label{esta}
Let $I \subseteq \mathbb R$ be any relatively compact open subset.
Then there exists $C>0$ such that uniformly in $z \in I \cup I_+$
\begin{align*}
  &|a| \leq C, \quad |b| \leq Cr^{\epsilon/2}, \quad {\rm Im}\, a \geq \frac{\epsilon}{2}|\nabla r|^2r^{-\epsilon/2-1}, \quad {\rm Im}\, b \geq \frac{\epsilon}{4}|\nabla r|^2r^{-1},\\ 
  &|\ell^{\bullet j}\nabla_j a| \leq Cr^{-\epsilon/2}f^{-1-\tau}, \quad |p^r b + b^2 - 2|\nabla r|^2(z-q_1+r^{\epsilon})| \leq Cf^{-1-{\rm min}\{ \rho, \epsilon', \tau \}}.
\end{align*}
\end{lem}

\begin{proof}
By the definitions of $a$ and $b$ (see \eqref{phasea}, \eqref{phaseb}) the first, second, third and fourth estimates clearly hold.
The fifth estimate is also clear by Condition~\ref{con2} and the following equation
\begin{align*}
\ell^{\bullet j}\nabla_j a &= \ell^{\bullet j}(\nabla_j \eta_{\lambda}|\nabla r|)r^{-\epsilon/2}\sqrt{2(z-q_1+r^{\epsilon})} \\
  &\quad - \frac{\epsilon}{2}(1-\eta)(\nabla r)^{\bullet}\eta_{\lambda}|\nabla r|r^{-\epsilon/2-1}\sqrt{2(z-q_1+r^{\epsilon})} \\
  &\quad + \frac{1}{2}\ell^{\bullet j}\eta_{\lambda}|\nabla r|r^{-\epsilon/2}\left( -(\nabla_j q_1) + \epsilon(\nabla r)_j r^{\epsilon-1} \right)/\sqrt{2(z-q_1+r^{\epsilon})} \\
  &\quad + i\ell^{\bullet j}\frac{\epsilon}{2}\left(\nabla_j |\nabla r|^2\right)r^{-\epsilon/2-1} - i\frac{\epsilon}{2}\left(\frac{\epsilon}{2}+1\right)(1-\eta)(\nabla r)^{\bullet}|\nabla r|^2r^{-\epsilon/2-2}.
\end{align*}
Since we can write
\begin{align*}
  &\quad \hspace{-3mm} p^rb + b^2 - 2|\nabla r|^2(z-q_1+r^{\epsilon}) \\
  &= (p^r \eta_{\lambda}|\nabla r|)\sqrt{2(z-q_1+r^{\epsilon})} + i\frac{\sqrt2}{2}\eta_{\lambda}|\nabla r|(\nabla^r q_1)/\sqrt{2(z-q_1+r^{\epsilon})} \\
  &\quad + i\eta_{\lambda}\frac{\sqrt2}{2}\epsilon|\nabla r|^3r^{-1}(z-q_1)/\sqrt{2(z-q_1+r^{\epsilon})} \\
  &\quad - 2(1-\eta_{\lambda}^2)|\nabla r|^2(z-q_1+r^{\epsilon}) + \frac{\epsilon}{4}(\nabla^r |\nabla r|^2)r^{-1} - \left( \frac{\epsilon}{4}+\frac{\epsilon^2}{16} \right)|\nabla r|^4r^{-2},
\end{align*}
by Condition~\ref{con2} the last estimate is also holds.
\end{proof}

\begin{lem} \label{Hz}
Let $I \subseteq \mathbb R$ be any relatively compact open subset.
Then there exist a complex-valued function $q_3$ and a constant $C>0$ such that uniformly in $z \in I \cup I_+$, as a quadratic forms on $C_0^{\infty}(\mathbb R^d)$,
$$
H-z = \frac{1}{2}(B+b)\tilde\eta(B-b) + \frac{1}{2}p_j\ell^{jk}p_k + q_3 ; \quad |q_3| \leq Cf^{-1-{\rm min}\{\rho, \epsilon', \tau \}}.
$$
\end{lem}

\begin{proof}
Using the following expression:
$$
B=p^r - \frac{i}{2}(\Delta r) = (p^r)^* + \frac{i}{2}(\Delta r),
$$
we can write
\begin{align*}
H-z &= \frac{1}{2}(B+b)\tilde\eta(B-b) + \frac{1}{2}p_j\ell^{jk}p_k - \frac{i}{2}(\nabla^r \tilde\eta)b + \frac{1}{2}\tilde\eta(p^r b) + \frac{1}{2}\tilde\eta b^2 \\
  &\quad - |x|^{\epsilon} + q_0 + q_1 + q_2 - z,
\end{align*}
where
\begin{align*}
q_0 &= \frac{1}{4}(\nabla^r \tilde\eta)(\Delta r) + \frac{1}{4}\tilde\eta(\nabla^r\Delta r) - \frac{1}{8}\tilde\eta(\Delta r)^2.
\end{align*}
Hence the assertion is obtained by setting
\begin{align*}
q_3 &= \frac{1}{2}\tilde\eta \left[ (p^r b) + b^2 - 2|\nabla r|^2(z-q_1+r^{\epsilon})\right] - (1-\eta)(z-q_1+r^{\epsilon}) \\
  &\quad - \frac{i}{2}(\nabla^r \tilde\eta)b + (r^{\epsilon}-|x|^{\epsilon}) + q_0 + q_2
\end{align*}
and using Lemma~\ref{esta}.
\end{proof}

Let us introduce the regularized weight
\begin{equation*}
\Theta= \Theta_{\nu}^{\delta}=\int_0^{f/R_{\nu}}(1+s)^{-1-\delta}ds=\left[ 1-(1+f/R_{\nu})^{-\delta} \right]/\delta; \ \delta>0, \nu\geq 0,
\end{equation*}
which is the same weight as \eqref{theta} introduced in Section 3.
We denote its derivatives in $f$ by primes such as \eqref{thetabibun}.

\begin{lem} \label{1708182102}
Let $I \subseteq \mathbb R$ be any relatively compact open subset,
and fix any $\delta \in (0, {\rm min}\{ \rho, \epsilon', \tau \}]$ and $\beta \in (0, 1+\epsilon/2)$.
Then there exist $c, C>0$ such that uniformly in $z \in I \cup I_+$ and $\nu \geq 0$, as quadratic forms on $\mathcal D(H)$
\begin{align*}
  &\quad \hspace{-3mm} {\rm Im}\left( (A-a)^*\Theta^{2\beta}(H-z) \right) \\
  &\geq c(A-a)^*\Theta'\Theta^{2\beta-1}(A-a) + cp_j\Theta^{2\beta}h^{jk}p_k \\
  &\quad - Cf^{-1-{\rm min}\{ 2\rho, 2\epsilon', 2\tau \}+2\delta}\Theta^{2\beta} - {\rm Re}\bigl( \gamma\Theta^{2\beta}(H-z) \bigr),
\end{align*}
where $\gamma$ is a certain function satisfying $|\gamma| \leq Cf^{-1-{\rm min}\{ 2\rho, 2\epsilon', 2\tau \}+2\delta}$.
\end{lem}

\begin{proof}
Let $I, \delta$ and $\beta$ be as in the assertion.
To prove the asserted inequality it suffices to compute as a quadratic forms on $C_0^{\infty}(\mathbb R^d)$.
By Lemma~\ref{thetaest}, \ref{Hz} and \ref{1708182102} and the Cauchy-Schwarz inequality it follows that uniformly in $z \in I \cup I_+$ and $\nu \geq 0$
\begin{equation}
\begin{split} \label{estHz}
  &\quad \hspace{-3mm} {\rm Im}\left( (A-a)^*\Theta^{2\beta}(H-z) \right) \\
  &= \frac{1}{2}{\rm Im}\left( (A-a)^*\Theta^{2\beta}(B+b)\tilde\eta r^{\epsilon/2}(A-a) \right) + \frac{1}{2}{\rm Im}\left( A\Theta^{2\beta}p_j\ell^{jk}p_k \right) \\
  &\quad - \frac{1}{2}{\rm Im}\left( a^*\Theta^{2\beta}p_j\ell^{jk}p_k \right) + {\rm Im}\left( (A-a)^*\Theta^{2\beta}q_3 \right) \\
  &= \frac{1}{2}(A-a)^*\beta\Theta'\Theta^{2\beta-1}\eta(A-a) - \frac{1}{4}(A-a)^*\Theta^{2\beta}(\nabla^r \tilde\eta)r^{\epsilon/2}(A-a) \\
  &\quad - \frac{\epsilon}{8}(A-a)^*\Theta^{2\beta}\eta r^{\epsilon/2-1}(A-a) + \frac{1}{2}(A-a)^*\Theta^{2\beta}\left( {\rm Im}\,b \right)\tilde\eta r^{\epsilon/2}(A-a) \\
  &\quad + \frac{1}{4}\left[ p_j\Theta^{2\beta}\ell^{jk}p_k, iA \right] + \frac{1}{2}{\rm Re}\left( A(1-\eta)\left(\Theta^{2\beta}\right)'p^f \right) \\
  &\quad - \frac{1}{2}{\rm Im}\left( a^*\Theta^{2\beta}p_j\ell^{jk}p_k \right) + {\rm Im}\left( (A-a)^*\Theta^{2\beta}q_3 \right) \\
  &\geq \frac{1}{2}(A-a)^*\left( \beta\Theta' - f^{-1-2\delta}\Theta \right)\Theta^{2\beta-1}(A-a) + \frac{1}{4}\left[ p_j\Theta^{2\beta}\ell^{jk}p_k, iA \right] \\
  &\quad - \frac{1}{2}{\rm Im}\left( a^*\Theta^{2\beta}p_j\ell^{jk}p_k \right) - C_1Q,
\end{split}
\end{equation}
where
$$
Q = f^{-1-{\rm min}\{ 2\rho, 2\epsilon', 2\tau\}+2\delta}\Theta^{2\beta} + p_jr^{-\epsilon}f^{-1-{\rm min}\{ 2\rho, 2\epsilon', 2\tau\}+2\delta}\Theta^{2\beta}\delta^{jk}p_k.
$$

Let us further estimate the terms on the right-hand side of \eqref{estHz}.
By Lemma~\ref{thetaest} the first term of \eqref{estHz} can be bounded as
\begin{equation} \label{estHz2}
\begin{split}
  &\quad \hspace{-3mm} \frac{1}{2}(A-a)^*\left( \beta\Theta' - f^{-1-2\delta}\Theta \right)\Theta^{2\beta-1}(A-a) \\
  &\geq c_1(A-a)^*\Theta'\Theta^{2\beta-1}(A-a) - C_2Q.
\end{split}
\end{equation}
To estimate the second term of \eqref{estHz} we use the following lemma used also in \cite{is}.
\begin{lem} \label{1708211321}
Let $\tilde g^{ij}=\tilde g^{ji} \in C^{\infty}(\mathbb R^d)$ for $i, j = 1, 2, \ldots d$.
Then, as a quadratic form on $C_0^{\infty}(\mathbb R^d)$,
\begin{align*}
\left[ p_i\tilde g^{ij}p_j, iA \right] &= p_i\left\{ \tilde g^{ij}\left(\nabla^2 f\right)_j{}^{k}+\left(\nabla^2 f\right)_j{}^i\tilde g^{kj}-\bigl(\nabla^f \tilde g\bigr)^{ik} \right\}p_k \\
  &\quad - {\rm Im}\left( \tilde g^{jk}(\nabla_k\Delta f)p_j \right).
\end{align*}
\end{lem}

\noindent
We apply Lemma~\ref{1708211321} with $\tilde g = \Theta^{2\beta}\ell$ to the second term of \eqref{estHz}.
Then we can estimate as follows.
\begin{equation} \label{estHz3}
\begin{split}
  &\quad \hspace{-3mm} \frac{1}{4}\left[ p_j\Theta^{2\beta}\ell^{jk}p_k, iA \right] \\
  &= \frac{1}{4}p_i\left\{ \Theta^{2\beta}\ell^{ij}\left(\nabla^2 f\right)_j{}^{k}+\left(\nabla^2 f\right)_j{}^i\Theta^{2\beta}\ell^{kj}-\bigl(\nabla^f \Theta^{2\beta}\ell\bigr)^{ik} \right\}p_k \\
  &\quad - \frac{1}{4}{\rm Im}\left( \Theta^{2\beta} \ell^{jk}(\nabla_k\Delta f)p_j \right) \\
  &\geq \frac{1}{2}p_j\left\{ h^{jk}\Theta - \beta r^{-\epsilon}\Theta'\ell^{jk} \right\}\Theta^{2\beta-1}p_k  - C_2Q.
\end{split}
\end{equation}
As for the third term of \eqref{estHz} using Lemma~\ref{esta} and the Cauchy-Schwarz inequality we can estimate as, for any $\epsilon_1 \in (0, 1)$,
\begin{equation} \label{estHz4}
\begin{split}
  &\quad \hspace{-3mm} - \frac{1}{2}{\rm Im}\left( a^*\Theta^{2\beta}p_j\ell^{jk}p_k \right) \\
  &= - \frac{1}{2}{\rm Im}\left( p_ja^*\Theta^{2\beta}\ell^{jk}p_k + i(\nabla_j a)^*\Theta^{2\beta}\ell^{jk}p_k +i2\beta(1-\eta)a^*\Theta^{2\beta-1}\Theta'p^f \right) \\
  &\geq \left( \frac{\epsilon}{4}-\epsilon_1 \right)p_jr^{-\epsilon/2-1}\Theta^{2\beta}\ell^{jk}p_k - \epsilon_1^{-1}C_3Q.
\end{split}
\end{equation}
By the bounds \eqref{estHz}, \eqref{estHz2}, \eqref{estHz3} and \eqref{estHz4} we obtain
\begin{equation} \label{estHz5}
\begin{split}
  &\quad \hspace{-3mm} {\rm Im}\left( (A-a)^*\Theta^{2\beta}(H-z) \right) \\
  &\geq c_1(A-a)^*\Theta'\Theta^{2\beta-1}(A-a) - \epsilon_1^{-1}C_4Q \\
  &\quad + \frac{1}{2}p_j\left\{ h^{jk}\Theta + \left( \frac{\epsilon}{2}-2\epsilon_1 \right)r^{-\epsilon/2-1}\ell^{jk}\Theta - \beta r^{-\epsilon/2}\Theta'\ell^{jk} \right\}\Theta^{2\beta-1}p_k.
\end{split}
\end{equation}
If we choose $\epsilon_1>0$ small enough, we have the following inequality
\begin{equation} \label{estHz6}
\begin{split}
  &\quad \hspace{-3mm} \frac{1}{2}p_j\left\{ h^{jk}\Theta + \left( \frac{\epsilon}{2}-2\epsilon_1 \right)r^{-\epsilon/2-1}\ell^{jk}\Theta - \beta r^{-\epsilon/2}\Theta'\ell^{jk} \right\}\Theta^{2\beta-1}p_k \\
  &\geq c_2p_j\Theta^{2\beta}h^{jk}p_k.
\end{split}
\end{equation}
Finally we can bound $-Q$ as
\begin{equation} \label{estHz7}
\begin{split}
-Q &\geq - C_5f^{-1-{\rm min}\{ 2\rho, 2\epsilon', 2\tau\}+2\delta}\Theta^{2\beta} \\
  &\quad - 2{\rm Re}\left( r^{-\epsilon}f^{-1-{\rm min}\{ 2\rho, 2\epsilon', 2\tau\}+2\delta}\Theta^{2\beta}(H-z) \right).
\end{split}
\end{equation}
By \eqref{estHz5}, \eqref{estHz6} and \eqref{estHz7}, if we set
$$
\gamma = 2\epsilon_1^{-1}C_4r^{-\epsilon}f^{-1-{\rm min}\{ 2\rho, 2\epsilon', 2\tau\}+2\delta},
$$
then the assertion follows.
\end{proof}

\subsection{Applications}

\begin{proof}[Proof of Theorem~\ref{rcb}]
Let $I \subset \mathbb R$ be any relative compact open subset.
For $\beta = 0$ the assertion is obvious by Theorem~\ref{bb}, and hence we may let $\beta \in (0, \beta_c)$.
We take any
$$
\delta \in (0, {\rm min}\{ \rho, \epsilon', \tau \}-\beta).
$$
By Lemma~\ref{1708182102}, the Cauchy-Schwarz inequality and the Theorem~\ref{bb} there exists $C_1>0$ such that for any state $\phi = R(z)\psi$ with $\psi \in r^{-\beta}\mathcal B$ and $z \in I_+$
\begin{equation} \label{rcbest}
\begin{split}
  &\quad \hspace{-3mm} \| \Theta'^{1/2}\Theta^{\beta-1/2}(A-a)\phi \|^2 + \langle p_j \Theta^{2\beta}h^{jk}p_k \rangle_{\phi} \\
  &\leq C_1\left[ \| \Theta^{\beta}(A-a)\phi \|_{\mathcal B^*}\| \Theta^{\beta}\psi \|_{\mathcal B} + \| f^{-1/2-{\rm min}\{\rho, \epsilon', \tau\}+\delta}\Theta^{\beta}\phi \|^2\right. \\
  &\quad \left. + \| f^{1/2-{\rm min}\{\rho, \epsilon', \tau\}+\delta}\Theta^{\beta}\psi \|^2 \right] \\
  &\leq C_2R_{\nu}^{-2\beta}\left[ \| f^{\beta}(A-a)\phi \|_{\mathcal B^*}\| f^{\beta}\psi \|_{\mathcal B} + \| f^{\beta}\psi \|_{\mathcal B}^2 \right].
\end{split}
\end{equation}
Here we note that $f^{\beta}(A-a)\phi \in \mathcal B^*$ for each $z \in I_+$ and hence the quantity on the right-hand side of \eqref{rcbest} is finite.
In fact, this can be verified by commuting $R(z)$ and powers of $f$ sufficiently many times and using the fact that $\psi \in f^{-\beta}\mathcal B$.
Then by \eqref{rcbest} it follows
\begin{equation} \label{rcbest2}
\begin{split}
  &\quad \hspace{-3mm} R_{\nu}^{2\beta}\| \Theta'^{1/2}\Theta^{\beta-1/2}(A-a)\phi \|^2 + R_{\nu}^{2\beta}\langle p_j \Theta^{2\beta}h^{jk}p_k \rangle_{\phi} \\
  &\leq C_2\left[ \| f^{\beta}(A-a)\phi \|_{\mathcal B^*}\| f^{\beta}\psi \|_{\mathcal B} + \| f^{\beta}\psi \|_{\mathcal B}^2 \right].
\end{split}
\end{equation}
In the first term on the right-hand side of \eqref{rcbest2} we take the supremum in $\nu \geq 0$ noting \eqref{thetabibun}, and then obtain
$$
c_1\| f^{\beta}(A-a)\phi \|_{\mathcal B^*}^2 \leq C_2\left[ \| f^{\beta}(A-a)\phi \|_{\mathcal B^*}\| f^{\beta}\psi \|_{\mathcal B} + \| f^{\beta}\psi \|_{\mathcal B}^2 \right],
$$
which implies
\begin{equation} \label{rcbest3}
\| f^{\beta}(A-a)\phi \|_{\mathcal B^*}^2 \leq C_3\| f^{\beta}\psi \|_{\mathcal B}^2.
\end{equation}
As for the second term on the right-hand side of \eqref{rcbest2} we use \eqref{rcbest3}, the concavity of $\Theta$ and Lebesgue's monotone convergence theorem, and then obtain by letting $\nu \to \infty$
$$
\langle p_j f^{2\beta}h^{jk}p_k \rangle_{\phi} \leq C_4\| f^{\beta}\psi \|_{\mathcal B}^2.
$$
Hence we are done.
\end{proof}

\begin{proof}[Proof of Corollary~\ref{lap}]
Let $s>1/2$ be as in the assertion.
Throughout the proof let us fix any $\beta \in (0, {\rm min}\{\beta_c, s-1/2 \})$ and $s' \in (s-\beta, s)$.
We decompose for $m \geq 0$ and $z, z' \in I_+$
\begin{equation} \label{lap1}
\begin{split}
R(z) - R(z') &= \chi_mR(z)\chi_m - \chi_mR(z')\chi_m \\
  &\quad + \bigl( R(z) - \chi_mR(z)\chi_m \bigr) - \left( R(z') - \chi_mR(z')\chi_m \right).
\end{split}
\end{equation}
By Theorem~\ref{bb} we can estimate the third term of \eqref{lap1} uniformly in $m \geq 0$ and $z, z' \in I_+$ as
\begin{equation} \label{lap2}
\begin{split}
  &\quad \hspace{-3mm} \| R(z) - \chi_mR(z)\chi_m \|_{\mathcal B(\mathcal H_s, \mathcal H_{-s})} \\
  &\leq \| f^{-s}\bar\chi_mR(z)\bar\chi_mf^{-s} \|_{\mathcal B(\mathcal H)} + \| f^{-s}\bar\chi_mR(z)\chi_mf^{-s} \|_{\mathcal B(\mathcal H)} \\
  &\quad + \| f^{-s}\chi_mR(z)\bar\chi_mf^{-s} \|_{\mathcal B(\mathcal H)} \\
  &\leq C_1R_m^{s'-s}.
\end{split}
\end{equation}
Similarly, we obtain
\begin{equation} \label{lap3}
\| R(z') - \chi_mR(z')\chi_m \|_{\mathcal B(\mathcal H_s, \mathcal H_{-s})} \leq C_2R_m^{s'-s}.
\end{equation}
As for the first and second terms on the right-hand side of \eqref{lap1}, using the equation
\begin{equation} \label{1709251357}
i[H, \chi_n] = {\rm Re}\left(\chi_n'p^f\right) = {\rm Re}\left(\chi_n'A\right)
\end{equation}
and noting the identify $\overline{a_{\bar z}}=a_z$, we can write for $n>m$
\begin{align*}
  &\quad \hspace{-3mm} \chi_mR(z)\chi_m - \chi_mR(z')\chi_m \\
  &= \chi_mR(z)\bigl\{ \chi_n(H-z') - (H-z)\chi_n \bigr\}R(z')\chi_m \\
  &= \chi_mR(z)\left\{ (z-z')\chi_n + i{\rm Re}\left(\chi_n'A\right) \right\}R(z')\chi_m \\
  &= \chi_mR(z)\left\{ (z-z')\chi_n - \frac{i}{2}(a_z - a_{z'})\chi_n' \right\}R(z')\chi_m \\
  &\quad + \frac{i}{2}\chi_mR(z)\chi_n'(A-a_{z'})R(z')\chi_m + \frac{i}{2}\chi_mR(z)(A+a_{\bar z})^*\chi_n'R(z')\chi_m.
\end{align*}
Then by Theorem~\ref{bb} and Theorem~\ref{rcb} we have uniformly in $n>m \geq 0$ and $z, z' \in I_+$
\begin{equation} \label{lap4}
\| \chi_mR(z)\chi_m - \chi_mR(z')\chi_m \|_{\mathcal B(\mathcal H_s, \mathcal H_{-s})} \leq C_3R_n|z-z'| + C_4R_m^{s'-s}.
\end{equation}
By \eqref{lap1}, \eqref{lap2}, \eqref{lap3} and \eqref{lap4}, we obtain uniformly in $n>m \geq 0$ and $z, z' \in I_+$
$$
\| R(z) - R(z') \|_{\mathcal B(\mathcal H_s, \mathcal H_{-s})} \leq C_5R_m^{s'-s} + C_3R_n|z-z'|.
$$
Now we choose $n = m+1$ and $R_m \leq |z-z'|^{-1/(s-s'+1)} \leq R_n$, and then obtain uniformly in $z, z' \in I_+$
\begin{equation} \label{lap5}
\| R(z) - R(z') \|_{\mathcal B(\mathcal H_s, \mathcal H_{-s})} \leq C_6|z-z'|^{\omega}
\end{equation}
with $\omega = (s-s')/(s-s'+1)$.
The H\"older continuity \eqref{Hc} for $R(z)$ follows from \eqref{lap5}.
The H\"older continuity \eqref{Hc} for $r^{-\epsilon/2}pR(z)$ follows by using \eqref{formula}.

The existence of the limits of \eqref{Rlim} follows immediately from \eqref{Hc}.
By Theorem~\ref{bb} the limits $R(\lambda \pm i0)$ and $r^{-\epsilon/2}pR(\lambda \pm i0)$ actually map into $\mathcal B^*$, and moreover they extend continuously to maps $\mathcal B \to \mathcal B^*$ by a density argument. 
Hence we are done.
\end{proof}

\begin{proof}[Proof of Corollary~\ref{rcb2}]
Note the elementary property
$$
\| \psi \|_{\mathcal B^*} = \sup_{n \geq 0}\| \chi_n\psi \|_{\mathcal B^*}; \quad \psi \in \mathcal B^*.
$$
Let $\beta \in [0, \beta_c)$ be as in the assertion.
By Theorem~\ref{rcb} there exists $C>0$ such that for any $\Gamma>0$ and $n > 0$
$$
\| \chi_nf^{\beta}(A-a)R(\lambda + i\Gamma)\psi \|_{\mathcal B^*} \leq C\| f^{\beta}\psi \|_{\mathcal B}, \quad \psi \in C_0^{\infty}(\mathbb R).
$$
By taking the limit $\Gamma \to 0$ and using Corollary~\ref{lap} and a density argument, we obtain
$$
\| \chi_nf^{\beta}(A-a)R(\lambda + i0)\psi \|_{\mathcal B^*} \leq C\| f^{\beta}\psi \|_{\mathcal B}, \quad \psi \in f^{-\beta}\mathcal B.
$$
Finally, by the Lebesgue's monotone convergence theorem we obtain
$$
\| f^{\beta}(A-a)R(\lambda + i0)\psi \|_{\mathcal B^*} \leq C\| f^{\beta}\psi \|_{\mathcal B}, \quad \psi \in f^{-\beta}\mathcal B.
$$
Similarly, we obtain
$$
\langle p_jf^{2\beta}h^{jk}p_k \rangle^{1/2}_{R(\lambda +i0)\psi} \leq C\| f^{\beta}\psi \|_{\mathcal B}.
$$
Hence we are done.
\end{proof}

\begin{proof}[Proof of Corollary~\ref{Sur}]
Let $\lambda \in \mathbb R$, $\phi \in \mathcal H_{\rm loc}$ and $\psi \in f^{-\beta}\mathcal B$ with $\beta \in [0, \beta_c)$.
We first assume $\phi = R(\lambda + i0)\psi$.
Then (i) and (ii) of the corollary hold by Corollaries~\ref{lap} and \ref{rcb2}.
Conversely, assume (i) and (ii) of the corollary, and let
$$
\phi' = \phi - R(\lambda + i0)\psi.
$$
Then by Corollaries~\ref{lap} and \ref{rcb2} it follows that $\phi'$ satisfies (i) and (ii) of the corollary with $\psi=0$.
In addition, we can verify $\phi' \in \mathcal B_0^*$ by the virial-type argument.
In fact noting the identity
$$
2{\rm Im}\bigl( \chi_{\nu}(H-\lambda) \bigr) = ({\rm Re}\,a)\chi_{\nu}' + {\rm Re}\bigl( \chi_{\nu}'(A-a) \bigr),
$$
cf. \eqref{A} and \eqref{1709251357}, we conclude that
\begin{equation} \label{Surineq}
0 \leq \langle ({\rm Re}\,a)\bar\chi_{\nu}' \rangle_{\phi'} = {\rm Re}\langle \chi_{\nu}'
(A-a) \rangle_{\phi'}.
\end{equation}
Taking the limit $\nu \to \infty$ and using $\phi' \in f^{\beta}\mathcal B^*$ and $(A-a)\phi' \in f^{-\beta}\mathcal B_0^*$ in \eqref{Surineq}, we obtain $\phi' \in \mathcal B_0^*$.
By Theorem~\ref{rell} it follows that $\phi'=0$.
Hence we have $\phi = R(\lambda + i0)\psi$.
\end{proof}

%%%%%%%%%%%%%%%%%%%%%%%%%%%%%%%%%%%%%%%%%%%%%%%%%%%%%%%%%%%%

\appendix
\section{Rellich's theorem}\label{appen}

In the proofs of Theorem~\ref{bb} and Corollary~\ref{Sur} the absence of $\mathcal B_0^*$-eigenfunctions for $H$ plays a major role.
This result was studied in \cite{i}.
However, the space $\mathcal B_0^*$ employed in \cite{i} is somewhat different from the one introduced in this paper.
Due to this, we actually need a slightly relaxed version of Rellich's theorem.

For comparison we set
\begin{align*}
\mathcal B_r^* &= \{ \psi \in \mathcal H_{\rm loc} \ | \ \|\psi\|_r < \infty \}, \quad \ \|\psi\|_r = \sup_{\nu \geq 0} R_{\nu}^{\epsilon/4-1/2}\| F(R_{\nu}\leq r \leq R_{\nu+1})\psi\|_{\mathcal H},
\\
\mathcal B_{r,0}^* &= \{\psi\in \mathcal B_r^*\ |\ \lim_{\nu \to \infty} R_{\nu}^{\epsilon/4-1/2}\| F(R_{\nu}\leq r \leq R_{\nu+1})\psi\|_{\mathcal H} = 0\}, \\
\mathcal B_f^* &= \{ \psi \in \mathcal H_{\rm loc} \ | \ \|\psi\|_f < \infty \}, \quad \ \|\psi\|_f = \sup_{\nu \geq 0} R_{\nu}^{-1/2}\| F(R_{\nu}\leq f \leq R_{\nu+1})\psi\|_{\mathcal H},
\\
\mathcal B_{f,0}^* &= \{\psi\in \mathcal B_f^*\ |\ \lim_{\nu \to \infty} R_{\nu}^{-1/2}\| F(R_{\nu}\leq f \leq R_{\nu+1})\psi\|_{\mathcal H} = 0\}.
\end{align*}
As we can see with ease, for $0<\epsilon<2$ the spaces $\mathcal B_r^*$ and $\mathcal B_f^*$ are the same, and for $\epsilon=2$ the following inclusion relations hold:
\begin{equation} \label{r-f}
\mathcal B_r^* \subsetneq \mathcal B_f^*, \quad \mathcal B_{r,0}^* \subsetneq \mathcal B_{f,0}^*.
\end{equation}
In \cite{i} we constructed a $\mathcal B_r^*$-eigenfunction.
Hence by \eqref{r-f} a $\mathcal B_f^*$-eigenfunction certainly exists.
Although the absence of $\mathcal B_{r,0}^*$-eigenfunctions was proved in \cite{i}, 
we use in this paper the absence of $\mathcal B_{f,0}^*$-eigenfunctions as follows.
\begin{thm}\label{rell}
Suppose Condition~\ref{con}, and let $\lambda \in {\mathbb R}$.
If a function $\phi \in \mathcal B_{f,0}^*$ satisfies that
$$
(H-\lambda)\phi =0,
$$
in the distributional sense, then $\phi =0$ in ${\mathbb R}^d$.
\end{thm}

We note that for $\epsilon=2$ we impose weaker assumption than \cite{i}, and by the second inclusion relation of \eqref{r-f} Theorem~\ref{rell} is stronger than \cite[Theorem~1.2]{i}.
As with \cite{i} we can prove Theorem~\ref{rell} using a commutator argument with some weight inside.
In the proof we need modification of weight function $\Theta$.
In fact, we need to replace the cut-off function $\chi_{m, n}(r)$ as $\chi_{m, n}(f)$.
Since the other points of the proof are almost the same, we omit the details.

%%%%%%%%%%%%%%%%%%%%%%%%%%%%%%%%%%%%%%%%%%%%%%%%%%%%%%%%%%%%

\end{document}